\documentclass[oribibl]{llncs}

\usepackage{hyperref}
\usepackage{amssymb}
\usepackage{amsmath}
\usepackage{graphicx}
\usepackage{epsfig}
\usepackage{color}
\usepackage{verbatim,subfigure,rotating,slashbox,multirow}
\usepackage{calc} 
\usepackage{url}

\newcommand{\lref}[1]{Lemma \ref{lem:#1}}

\newcommand{\cref}[1]{Corollary \ref{cor:#1}}

\newcommand{\sref}[1]{Section \ref{sec:#1}}

\newcommand{\fref}[1]{Fig. \ref{fig:#1}}

\newcommand\set[1]{\{ #1\}}

\newcommand{\DTC}{\texttt{DTC}}
\newcommand{\threesat}{\texttt{3SAT}}
\newcommand{\SAT}{\texttt{SAT}}
\newcommand{\pthreesat}{\texttt{P3SAT}}
\newcommand{\bigO}{\operatorname{O}}

\newcommand{\NOT}{\texttt{NOT}}
\newcommand{\AND}{\texttt{AND}}

\newcommand{\true}{\texttt{T}}
\newcommand{\false}{\texttt{F}}

\begin{document}

\title{Domino Tatami Covering is NP-complete}
%


\author{Alejandro Erickson \and  Frank Ruskey}

\institute{Department of Computer Science, University of Victoria, V8W 3P6, Canada}

\maketitle

\begin{abstract}
  A covering with dominoes of a rectilinear region is called
  \emph{tatami} if no four dominoes meet at any point.  We describe a
  reduction from planar~\threesat~to Domino Tatami Covering.  As a
  consequence it is NP-complete to decide whether there is a perfect
  matching of a graph that meets every 4-cycle, even if the graph is
  restricted to be an induced subgraph of the grid-graph.  The gadgets
  used in the reduction were discovered with the help of a
  \SAT-solver.
\end{abstract}

\section{Introduction}
\label{sec:intro}

Imagine that you want to ``pave'' a rectilinear driveway on the
integer lattice using 1 by 2 bricks.  Sometimes this will be possible,
but sometimes not, depending on the shape of the driveway.
Abstractly, a rectilinear driveway $D$ is a connected finite induced
subgraph of the infinite planar grid-graph, and a paving with bricks
corresponds to a perfect matching.  Since $D$ is bipartite, various
network flow algorithms can be used to determine whether there is a
paving in low-order polynomial time.

However, an examination of typical paving patterns reveals that
another constraint is often enforced/desired, probably for both
aesthetic reasons and engineering reasons.  The constraint is that no
four bricks meet at a point.  In some recent papers, this restriction
has come to be known as the tatami constraint, because Japanese tatami
mat layouts often adhere to it.  The question that we wish to address
in this paper is: What is the complexity of determining whether $D$
has a paving also satisfying the tatami constraint?  We will show that
the problem is NP-complete.

A \emph{rectilinear region} is a polyomino which may have holes.  We
describe a polynomial reduction from the NP-complete problem
planar~\threesat~to Domino Tatami Covering (\DTC).  The gadgets used in
the reduction were discovered with the help of a \SAT-solver.

\begin{definition}[Domino Tatami Covering (\DTC)]
  \label{def:dtc}
  ~
  \begin{description}
  \item[INSTANCE:] A rectilinear region $R$, on the integer lattice,
    represented, say, as $n$ line segments joining the corners of the
    polygon.
  \item[QUESTION:] Can $R$ be covered by dominoes such that no four of
    them meet at any one point?
  \end{description}
\end{definition}



Domino tatami coverings have an interesting combinatorial structure,
elucidated for rectangles in \cite{RuskeyWoodcock2009} and further in
\cite{EricksonRuskeySchurch2011}.  The results in these papers, as
well as \cite{EricksonSchurch2012,EricksonRuskey2013a} are
enumerative, whereas in this paper we explore tatami coverings from a
computational perspective.  There is no comprehensive structure
theorem for tatami coverings of rectilinear grids, but evidently much
of the structure is still there, as is illustrated in \fref{tatami}.

\begin{figure}[!ht]
  \centering
  \includegraphics[width=\textwidth]{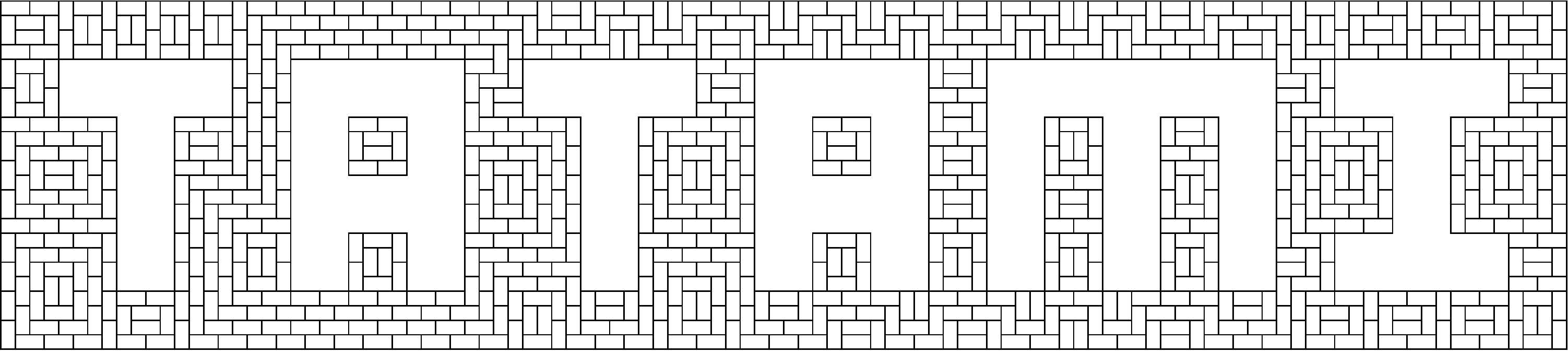}
  \caption{A domino tatami covering of a rectilinear region.  This particular covering was
  produced by a \SAT-solver.}
  \label{fig:tatami}
\end{figure}

There are some previous complexity results about tilings and domino coverings.
Historically, perhaps the first concerned colour-constrained coverings, such as those of
Wang tiles.  It is well known, for example, that covering the $k\times
k$ grid with Wang tiles is NP-complete (\cite{Lewis1978}).  On the
other hand tatami does not appear to be a special case of these, nor
of similar colour restrictions on dominos (e.g. \cite{Biedl2005,WormanWatson2004}).

A more closely related mathematical context is found, instead, among the graph
matching problems discussed by Churchley, Huang, and Zhu, in
\cite{ChurchleyHuangZhu2011}.  In their paper, an \emph{$H$-transverse
  matching} of a graph $G$, is a matching $M$, such that $G-M$ has no
subgraph $H$.  In a tatami covering of the rectilinear grid, $G$ is a
finite induced subgraph of the infinite grid-graph, $H$ is a $4$-cycle, and we
require a perfect matching of the edges.  In fact, if the matching is
not required to be perfect, the problem is polynomial.

\SAT-solvers have been applied to a broad range of industrial and
mathematical problems in the last decade.  Our reduction from
planar~\threesat~uses Minisat (\cite{EenSorensson2004}) to help automate
gadget generation, as was also done by Ruepp and Holzer
(\cite{RueppHolzer2010}).  It is easy to see that instances of other
locally restricted covering problems can be expressed as
satisfiability formulae, which suggests that \SAT-solvers may provide a
methodological applicability in hardness reductions involving those problems.


\section{Preliminaries}
\label{sec:preliminaries}

Let $\phi$ be a CNF formula, with variables $U$, and clauses $C$. The
formula is \emph{planar} if there exists a planar graph $G(\phi)$ with
vertex set $U\cup C$ and edges $\set{u,c}\in E$, where one of the
literals $u$ or $\bar u$ is in the clause $c$.  When the clauses
contain at most three literals, $\phi$ is an instance of \pthreesat,
which is NP-complete (\cite{Lichtenstein1982}).

We construct an instance of \DTC~which emulates a given instance,
$\phi$, of \pthreesat, by replacing the parts of $G(\phi)$ with a
rectilinear region, $R(\phi)$, that can be tatami-covered with
dominoes if and only if $\phi$ is satisfiable.  Let $n=|U\cup C|$.  In
\sref{layout} we show that $R(\phi)$ can be created in $\bigO(n)$
time, and that it fits in a $\bigO(n) \times \bigO(n)$ grid, by using
Rosenstiehl and Tarjan's algorithm (\cite{RosenstiehlTarjan1986}).


\section{Gadgets}
\label{sec:gadgets}


In this section we describe wire, \NOT~gates, and \AND~gates, which form the required
gadgets.  The functionality of our gadgets depends on the coverings of
a certain $8\times 8$ subgrid.

\begin{lemma}
  \label{lem:square}
  Let $R$ be a rectilinear grid, with an $8\times 8$ subgrid, $S$.  If
  a domino crosses the boundary of $S$ in a domino tatami covering of
  $R$, then at least one corner of $S$ is also covered by a domino
  that crosses its boundary.
\end{lemma}
\begin{proof}
  Suppose $R$ is covered by dominoes, and consider those dominoes
  which cover $S$.  Such a cover may not be exact, in the sense that a
  domino may cross the boundary of $S$.  If we consider all such
  dominoes to be monominoes within $S$, we obtain a monomino-domino
  covering of $S$.  This covering inherits the tatami restriction from
  the covering of $R$, so it is one of the $8\times 8$ monomino-domino
  coverings enumerated in \cite{EricksonSchurch2012} (and/or
  \cite{EricksonRuskeySchurch2011}).

  The proof of Lemma 4 in \cite{EricksonSchurch2012} (third paragraph)
  states that there is a monomino in at least one corner of $S$ if
  $0<m<n$; Corollary 2 of \cite{EricksonRuskeySchurch2011} states that
  there is a monomino in at least one corner of $S$ if $m=n$ (see
  examples in \fref{monominoSquare}).  This monomino corresponds to a
  domino which crosses the boundary in a corner of $S$, as
  required.\hfill \ensuremath{\Box}
\end{proof}

\begin{figure}[h]
  \centering
  \includegraphics{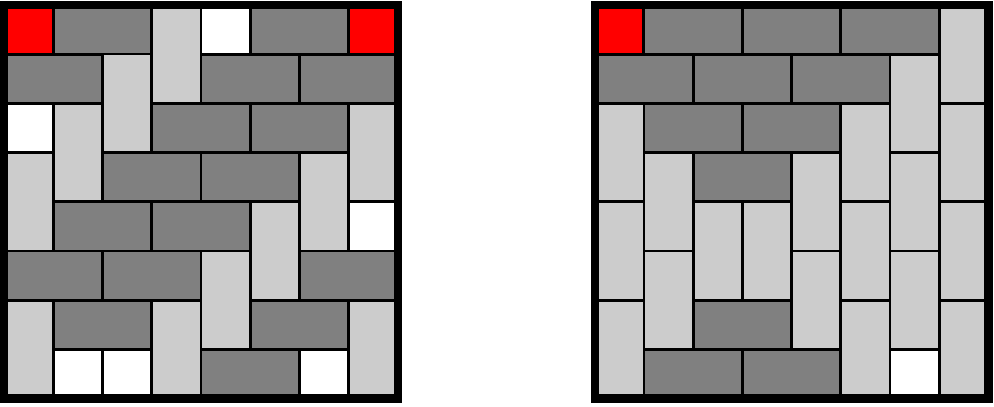}
  \caption{All monomino-domino tatami coverings of the square have at
    least one monomino in their corners
    (see~\cite{EricksonRuskeySchurch2011,EricksonSchurch2012}).  The
    squares in $R(\phi)$ have isolate corners, so these must be
    covered in exactly one of the two ways given by Exercise 7.1.4.215
    in \cite{Knuth2011}, shown in \fref{notgate}.}
  \label{fig:monominoSquare}
\end{figure}

The rectilinear region $R(\phi)$ incorporates a network of $8\times 8$
squares, whose centres reside on a $16\mathbb{Z}~\times~16\mathbb{Z}$
grid, and whose corners form part of the boundary of $R(\phi)$.
\lref{square} implies that no domino may cross their boundaries, and
thus each one must be covered in one of the two ways shown in
\fref{notgate}. (For proofs see \cite{RuskeyWoodcock2009} and Exercise 215, Section
7.1.4 in \cite{Knuth2011}).

The coverings of these squares are related to each other by connecting
regions.  The part of an $8\times 8$ square which borders on a
connector may be covered either by two tiles, denoted by \false~to
signify ``false'', or three tiles, denoted by \true~to signify
``true'' (see \fref{notgate}).  Note that the covering of a square is
not \true~or \false~by itself, because connectors below and beside it
would meet the square at differing interfaces.

A connector, which imposes a relationship between the coverings of a
set of $8\times 8$ squares, is verified by showing that it can be
covered if and only if the relationship is satisfied.  The connectors
we describe were generated with \SAT-solvers, but they are simple enought that
we can verify them by hand, as is done below.

\paragraph{\NOT~gate.} The \NOT~gate interfaces with two $8\times 8$
squares (see \fref{notgate}), and can be covered if and only if these
squares are covered with differing configurations.

\begin{figure}[ht]
  \centering %
  \scriptsize \def\svgscale{0.5} %
  \subfigure[\NOT~gate with \texttt{F} and \texttt{T}
  interfaces.]{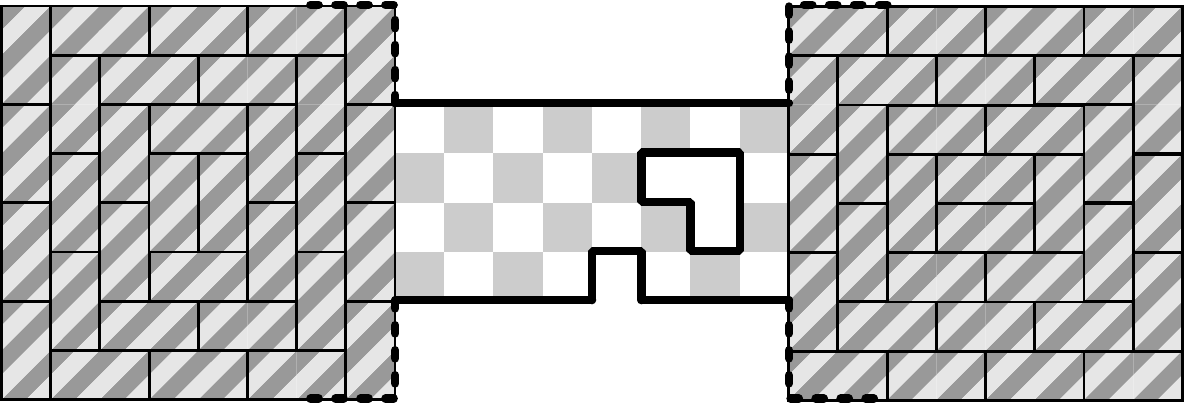%
    \label{fig:notgate}}%
  \subfigure[\texttt{F}$\longrightarrow$\texttt{T}.]{\includegraphics[scale=.5]{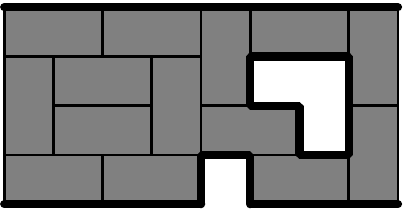}%
    \label{fig:notgateFT}}%
  \hspace{0.25in}%
  \subfigure[\texttt{T}$\longrightarrow$\texttt{F}.]{\includegraphics[scale=.5]{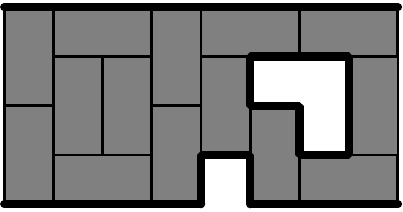}%
    \label{fig:notgateTF}}%
  \hspace{0.25in}%
  \def\svgscale{0.5}  %
  \subfigure[\texttt{F}$\longrightarrow$\texttt{F}.]{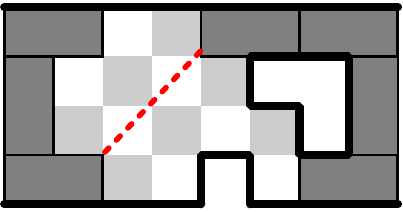%
    \label{fig:notgateFF}}%
  \hspace{0.25in}%
\def\svgscale{0.5}  %
  \subfigure[\texttt{T}$\longrightarrow$\texttt{T}.]{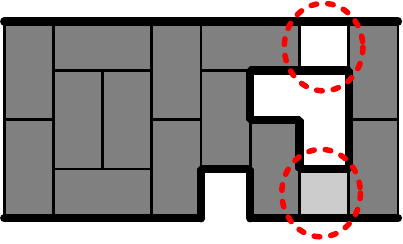%
    \label{fig:notgateTT}}%
  \caption{\NOT~gate can be covered if and only if the input differs
    from the output.  Numbered tiles indicate the (non-unique)
    ordering in which their placement is forced.  Red dotted lines
    indicate how domino coverings are impeded in \emph{(d)} and
    \emph{(e)}.}
  \label{fig:not}
\end{figure}

\paragraph{Wire gadget.} Wire transmits \true~or \false~through a
sequence of squares (see \fref{wireunit}).  A turn may incorporate a
\NOT~gate in order to maintain the same configuration (see
\fref{wireturn}).

\begin{figure}[h!]
  \centering %
  \scriptsize %
  \subfigure[Unit of wire, carrying \texttt{T}.]{\includegraphics[scale=.5]{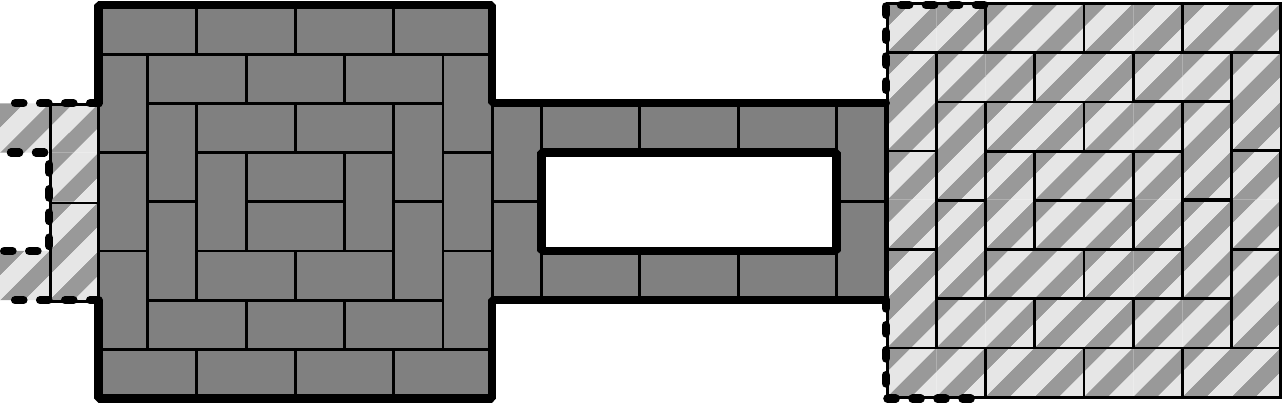}
    \label{fig:wireunit}}%

  \def\svgscale{0.5} %
  \subfigure[Wire branch and turn, carrying
  \texttt{F}.]{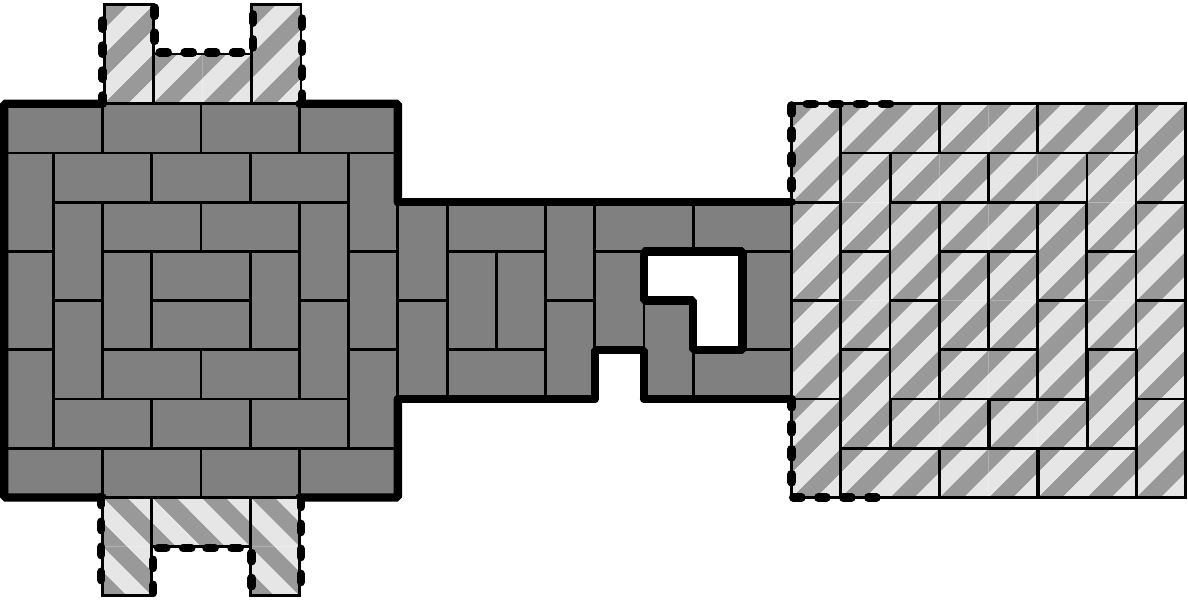
    \label{fig:wireturn}}%
  \caption{Wire gadget.}
  \label{fig:wire}
\end{figure}

\paragraph{AND gate.} The \AND~gate interfaces with two $8\times 8$
input squares, and one output square (see \fref{andgate}).  It can be
covered with dominoes if and only if the output value is the \AND~of
the inputs (see Figs.~\ref{fig:andgood} and~\ref{fig:andbad}).

\begin{figure}[h!]
  \centering %
  \scriptsize %
  \def\svgscale{0.5} %
  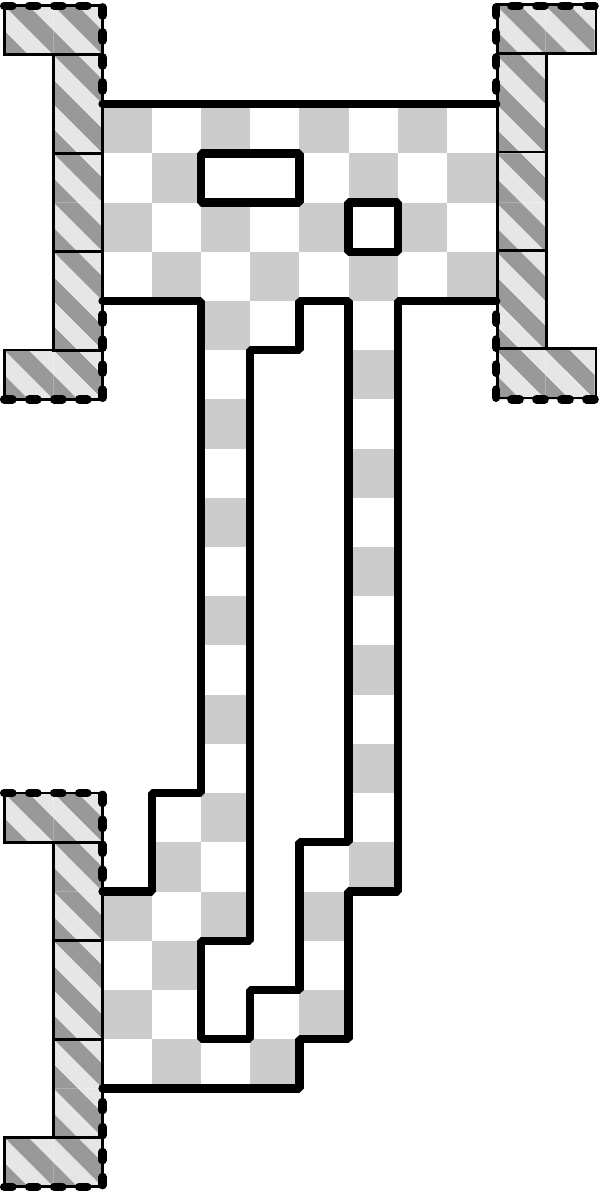
  \caption{\AND~gate with input (\texttt{T},\texttt{T}).}
  \label{fig:andgate}
\end{figure}

\begin{figure}[h!]
  \centering %
  \scriptsize %
  \def\svgscale{0.3} %
  \subfigure[\texttt{T}\texttt{T}$\longrightarrow$\texttt{T}.]{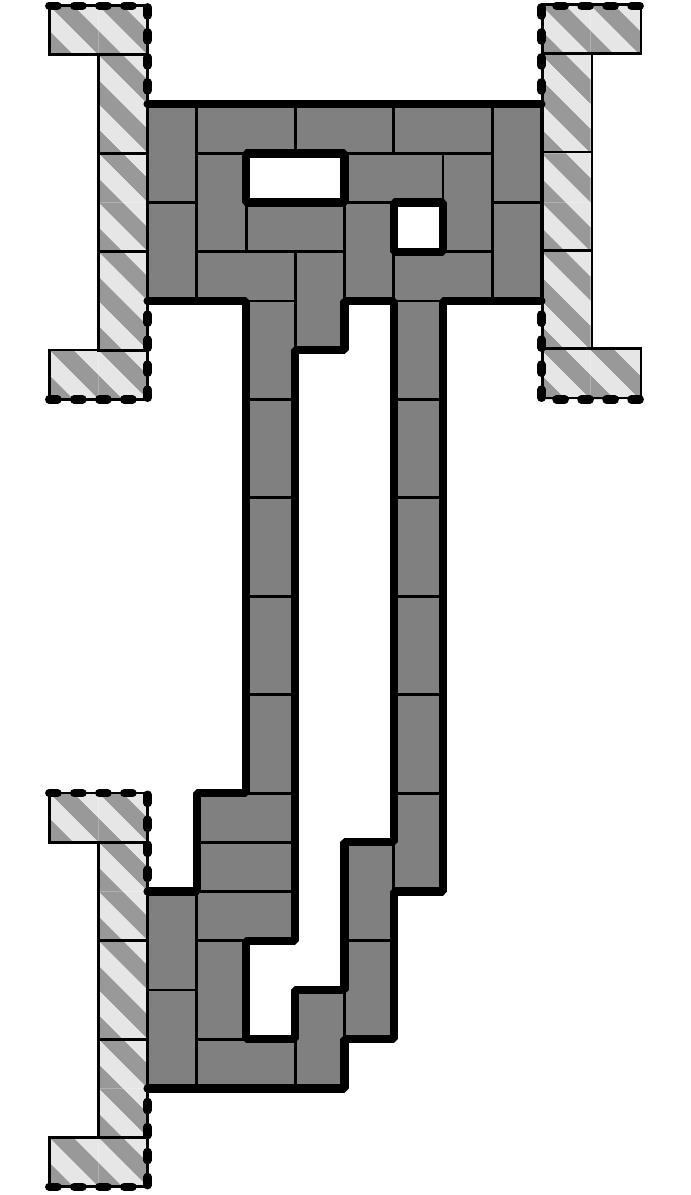%
    \label{fig:andTT}}%
  \hspace{0.2in} %
  \def\svgscale{0.3} %
  \subfigure[\texttt{T}\texttt{F}$\longrightarrow$\texttt{F}.]{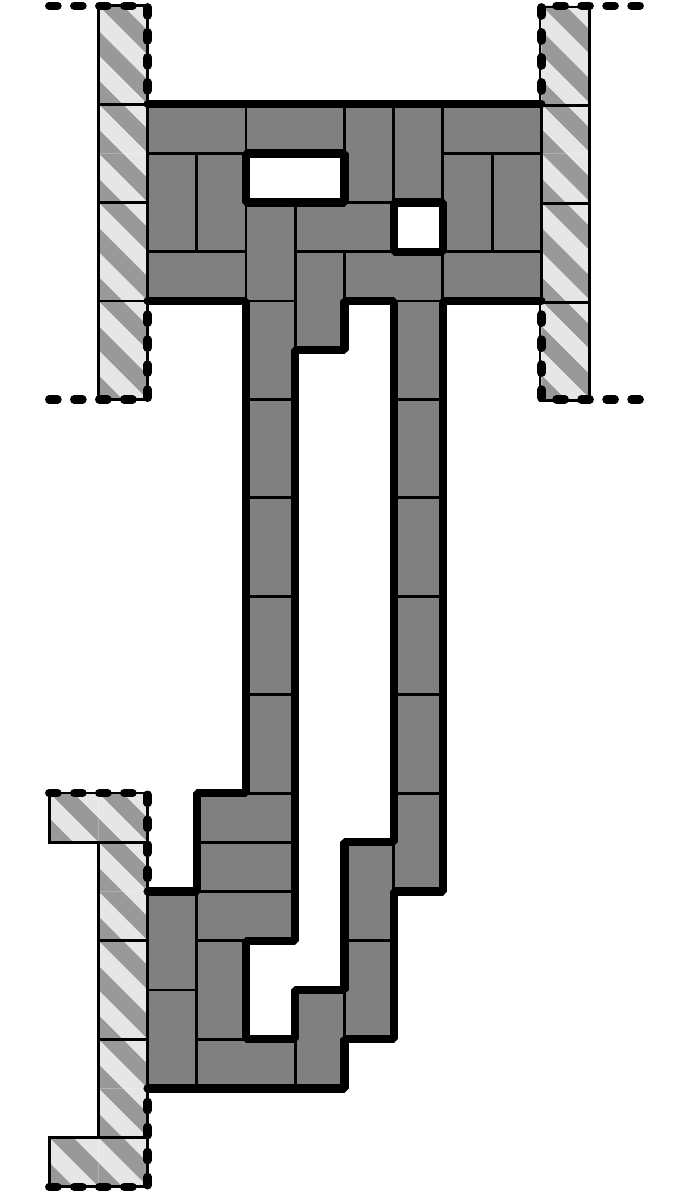%
    \label{fig:andTF}}%
  \hspace{0.2in} %
  \def\svgscale{0.3} %
  \subfigure[\texttt{F}\texttt{T}$\longrightarrow$\texttt{F}.]{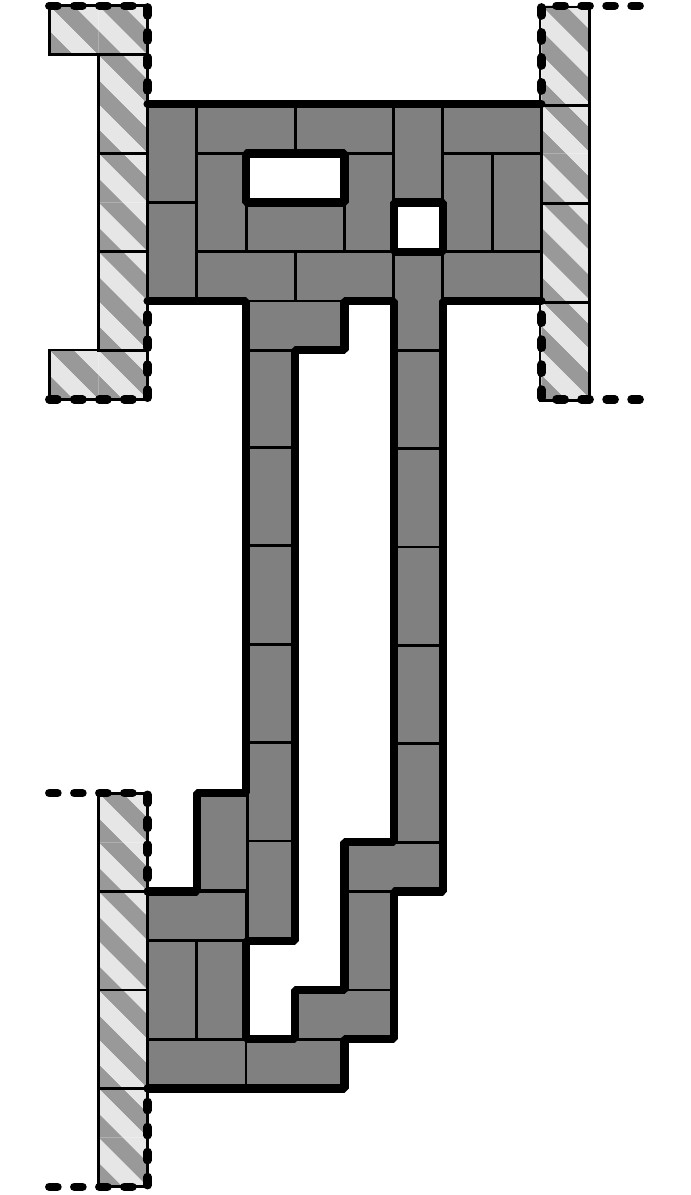%
    \label{fig:andFT}}%
  \hspace{0.2in} %
  \def\svgscale{0.3} %
  \subfigure[\texttt{F}\texttt{F}$\longrightarrow$\texttt{F}.]{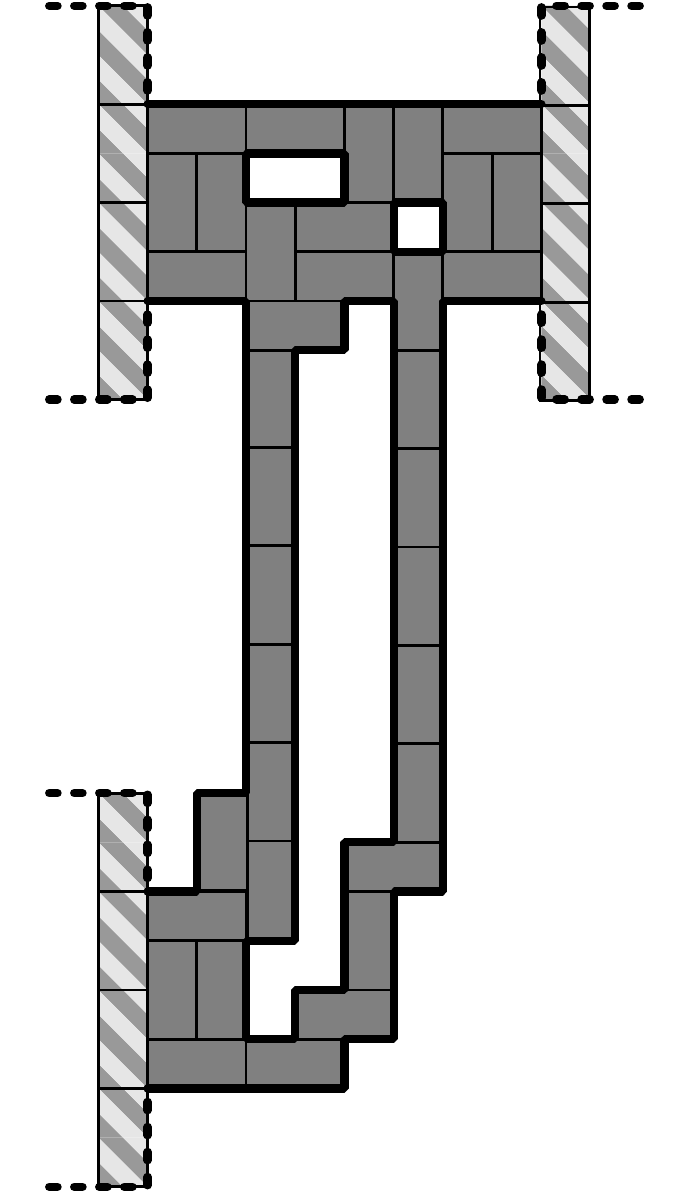%
    \label{fig:andFF}}%

  \caption{\AND~gate coverings.}
  \label{fig:andgood}
\end{figure}

\begin{figure}[h!]
  \centering %
  \scriptsize %
  \def\svgscale{0.4} %
  \subfigure[\texttt{*}\texttt{F}$\longrightarrow$\texttt{T}.]{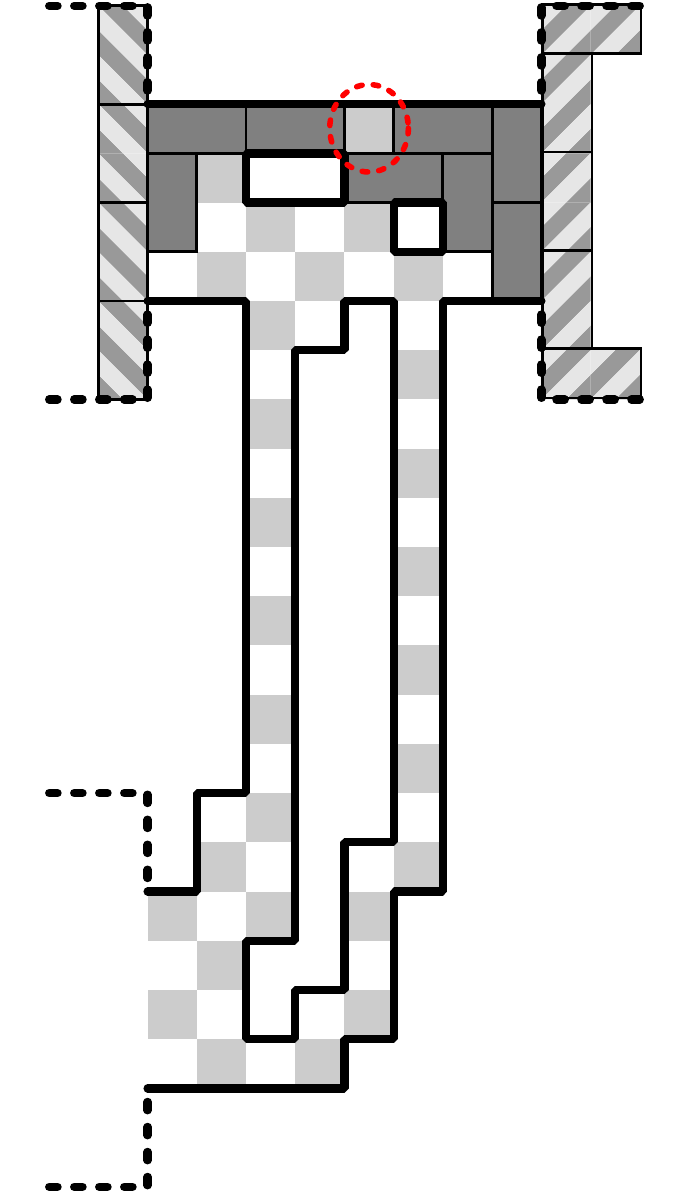%
    \label{fig:andsFT}}%
  \hspace{0.2in} %
  \def\svgscale{0.4} %
  \subfigure[\texttt{F}\texttt{*}$\longrightarrow$\texttt{T}.]{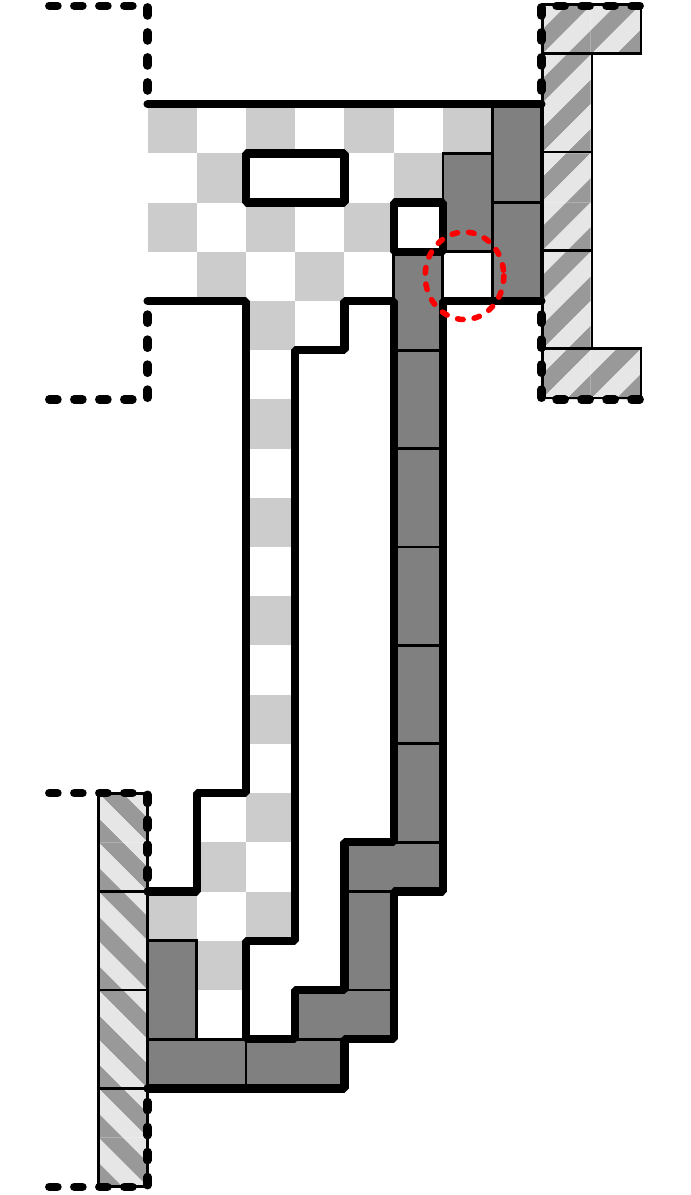%
    \label{fig:andFsT}}%
  \hspace{0.2in} %
  \def\svgscale{0.4} %
  \subfigure[\texttt{T}\texttt{T}$\longrightarrow$\texttt{F}.]{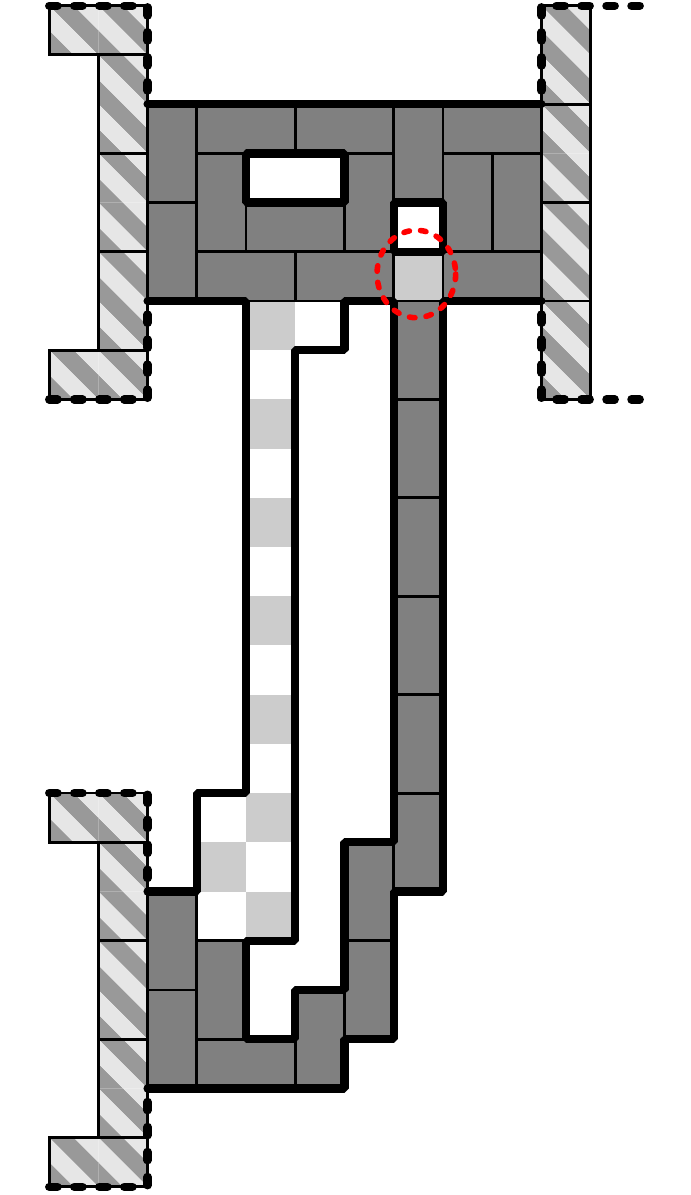%
    \label{fig:andTTF}}%
  \caption{Impossible \AND~gate coverings, where \texttt{*} denotes
    \texttt{F} or \texttt{T}.}
  \label{fig:andbad}
\end{figure}

\paragraph{Variable gadget.}  We use a vertical segment of wire.  The
variable gadget is set to \true~or \false~by choosing the appropriate
covering of one of its $8\times 8$ squares.  Its value (or its
negation) is propagated to clause gadgets via horizontal wire gadgets,
representing edges.


\paragraph{Clause gadget.} The clause gadget is a circuit for $\neg
(\bar a \wedge (\bar b \wedge \bar c))$, or the equivalent with fewer
inputs, ending in a configuration that can be covered if and only if
the output signal of the circuit is \true.  To satisfy the layout
requirements, the inputs to the clause are vertically translated by
wire (see \fref{clausestuff}).

\begin{figure}[!ht]
  \centering \def\svgscale{0.1} %
  \subfigure[]{%
    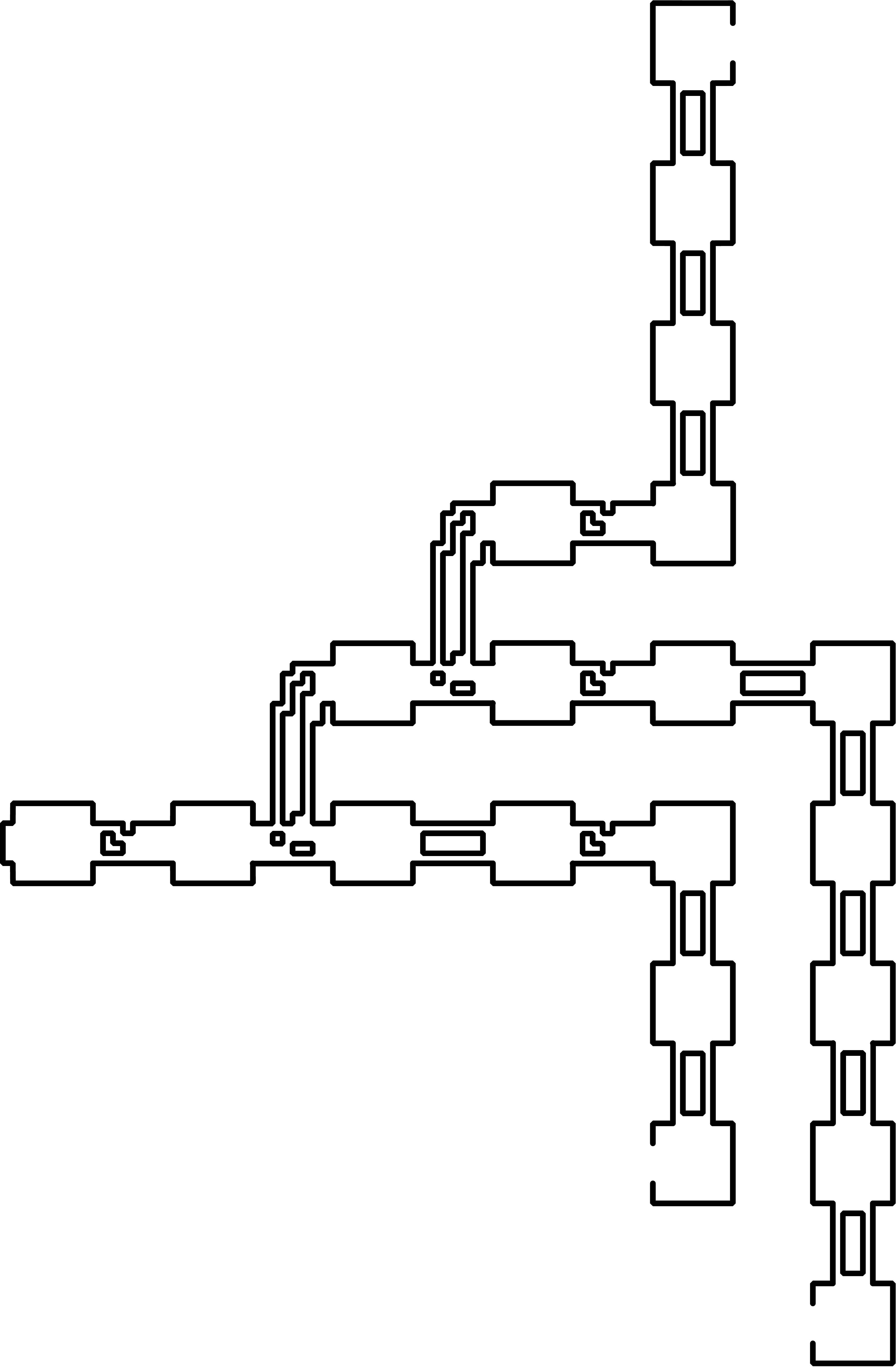%
    \label{fig:clause}}%
  \hspace{0.5in} %
  \def\svgscale{0.7}%
  \subfigure[End of clause.]{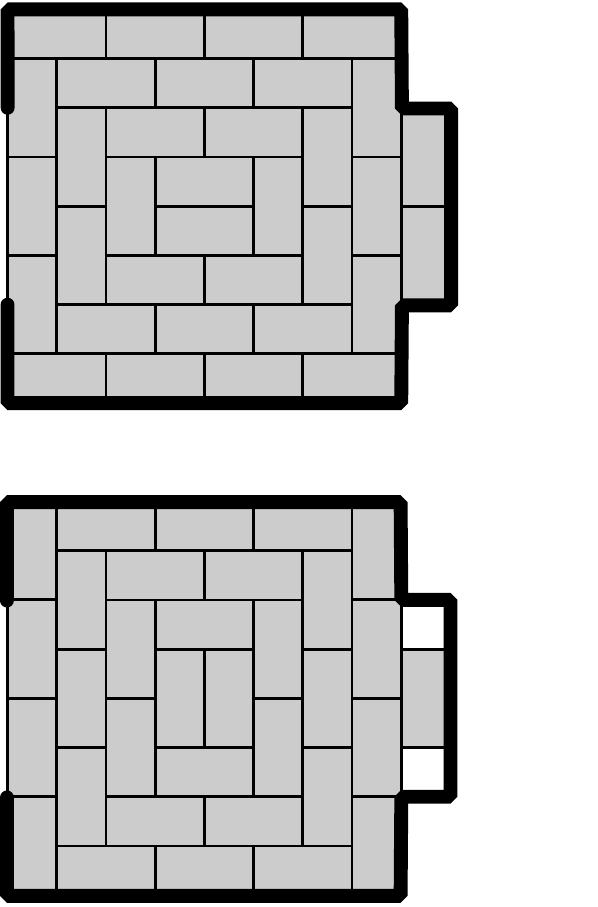%
     \label{fig:endclause}}%
   \caption{A three input clause gadget from the circuit $\neg (\bar a
     \wedge (\bar b \wedge \bar c))$.  Vertical wire translates
     horizontal inputs without changing the signal.  The end of the
     clause is coverable if and only if its signal is \true.}
  \label{fig:clausestuff}
\end{figure}

\section{Layout}
\label{sec:layout}

Let $G(\phi)$ be a planar embedding of the Boolean 3CNF formula
$\phi$, using Rosenstiehl and Tarjan's (\cite{RosenstiehlTarjan1986})
algorithm, so that each vertex is represented by a vertical line
segment, and each edge is represented by a horizontal line segment.
All parts lie on integer grid lines, inside of a
$\bigO(n)\times\bigO(n)$ grid, where $n = |U\cup C|$, and the
embedding is found in $\bigO(n)$ time.

There exists a constant $K$, which is the same for any planar 3CNF
formula, such that $G(\phi)$ can be scaled to fit on the $nK\times nK$
grid, and its parts replaced by the gadgets described above.  This
ensures that $R(\phi)$ has $\bigO(n^2)$ corners, and can also be
created in $\bigO(n)$ time.

The variable gadget is connected to edges by branches.  The layout of
$G(\phi)$ prevents conflicts between edges meeting the variable gadget
on the same side, while two edges can meet the left and right sides of
the variable gadget without interfering with each other.  The inputs
of the clause gadget are symmetric, so there are no conflicts when
connecting these to horizontal edges (see \fref{clause}).

\paragraph{Example.} The planar Boolean formula from Figure 1 in
\cite{Lichtenstein1982} gives the \DTC~instance in \fref{example}.

\begin{figure}[!ht]
  \centering
  \def\svgscale{0.15}
  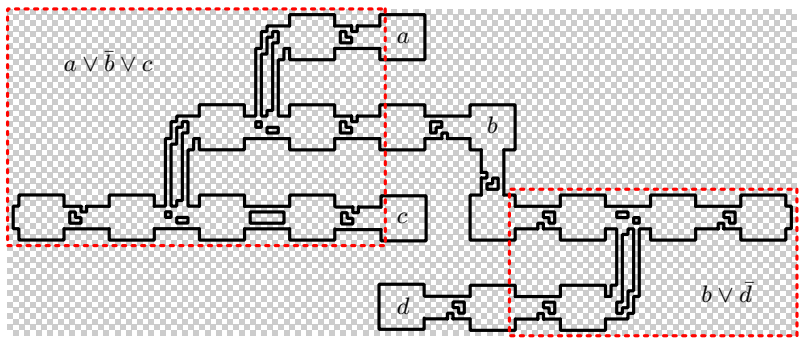
  \caption{An instance of \DTC~for the formula $(a\vee\bar b\vee
    c)\wedge (b\vee \bar d)$.}
  \label{fig:example}
\end{figure}

\section{\SAT~solver}
\label{sec:SAT}

The search for logical gates required fast testing of small
\DTC~instances.  We reduced \DTC~to \SAT~in order to use the
\SAT~solver, Minisat (\cite{EenSorensson2004}), and efficiently test
candidate regions connecting $8\times 8$ squares while satisfying the
conditions of the desired gate.  The \DTC~solver was also allowed to
make certain decisions about the region, rather than simply testing
regions generated by another program.

Our search algorithm requires the following inputs:
\begin{itemize}
\item an $r\times c$ rectangle of grid squares, partitioned into
  pairwise disjoint sets $K,X,A,C$; and,
\item a set of partial (good) coverings, $G$, and partial (bad)
  coverings, $B$, of $C$.
\end{itemize}

The output, $R$, is the region $A'\cup K$, where $A'\subseteq A$,
which satisfies the following constraints.
\begin{itemize}
\item[\emph{(g)}] There exist coverings of $R$ which form partial tatami domino
  coverings with each element of $G$.
\item[\emph{(b)}] There exists no covering of $R$ which forms a partial tatami
  domino covering with an element of $B$.
\end{itemize}

The outer loop of the search algorithm calls the \SAT-solver to find a
region that satisfies all elements of $G$, and avoids a list of
forbidden regions, which is initially empty. Upon finding such a
region, the inner loop checks whether the region satisfies any element
of $B$.  The search succeeds when \emph{(g)} and \emph{(b)} are
both satisfied, and fails if the outer loop's \SAT~instance has no
satisfying assignment.

The search space grows very quickly for several reasons, not least of
which is the fact that $2^{160}$ regions are possible within the
$20\times 8$ rectangle occupied by our \AND~gate (if corners are
allowed to meet one another).  In addition, the list of forbidden
regions, $L$, becomes too large for the \SAT~solver to handle
efficiently.

We used two heuristics on the inputs to obtain a feasible search.  The
first was searching for a smaller \AND~gate, which we modified to fit
the placement of the~$8\times 8$ squares.  The second was choosing
forbidden squares,~$X$, and required squares,~$K$, to reduce the
number of trivially useless regions that are tested.

\subsection{\DTC~as a Boolean formula}
The \SAT~instances used above are modifications of a formula which is
satisfiable if and only if a given region has a domino tatami
covering.

Let $R$ be the region we want to cover, and consider the graph whose
vertices are the grid squares of $R$, and whose edges connect vertices
of adjacent grid squares.  Let $H$ be the set of horizontal edges and
let $V$ be the set of vertical edges.  The variables of the
\SAT~instance are $H\cup V$, and those variables set to true in a
satisfying assignment are the dominoes in the covering. The clauses
are as follows, where $h,h'\in H$ and $v,v'\in V$.
\begin{enumerate}
\item Ensure a matching: For each pair of incident horizontal edges
  $(h,h')$, require the clause $\bar h \vee \bar h'$, and similarly
  for $(v,v')$, $(h,v)$.
\item Ensure the matching is perfect: For each set of edges
  $\set{h,h',v,v'}$, which are incident to a vertex, require the
  clause $h\vee h'\vee v\vee v'$.
\item Enforce the tatami restriction: For each 4-cycle, $hvh'v'$,
  require the clause $h\vee h'\vee v\vee v'$.
\end{enumerate}

\section{Variations and future work}
\label{sec:conclusion}




There are other locally constrained covering problems that are easily
represented as Boolean formulae.  Some of these are obviously
polynomial, such as monomino-domino tatami covering, but others may be
NP-complete.  \SAT-solvers can sometimes be used in such problems to create
elaborate gadgets, which may help find a hardness reduction.

An example problem, whose computational complexity is open, is Lozenge Tatami
Covering.  This problem is the decision about whether or not a finite
sub-grid of the triangular lattice can be covered with lozenges, such
that no more than $4$ lozenges meet at any point.  A structure similar
to that of tatami coverings occurs for this constraint (see
\fref{lozenge}).


\begin{figure}[!ht]
  \centering
\includegraphics[scale=0.5]{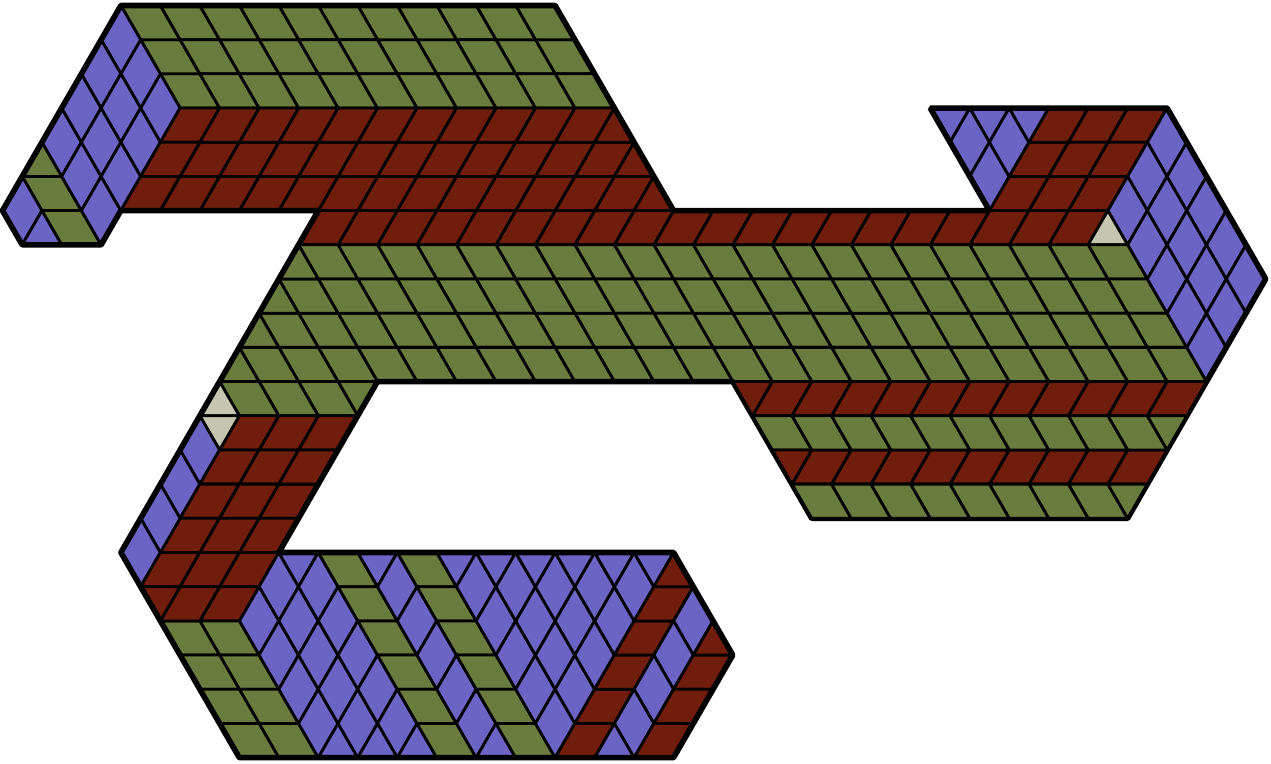}
  \caption{A triangle-lozenge tatami covering.}
  \label{fig:lozenge}
\end{figure}

Our main question about \DTC~is the complexity of the case where the
region is simply connected (no holes).  We believe that the problem is still
NP-complete, but a completely new approach will be required.

Secondarily, we are interested
in $H$-transverse perfect matchings for $H$ and $G$ other than $C_4$
and grid-graphs.  Are there other $H$-transverse perfect matchings of
interest which induce a tatami-like global structure in the containing
graph?

Another variant, mildly advocated by Don Knuth (personal communication), concerns inner corners of the 
coverings, such as occurs at the upper left in the letter \texttt{T} in
Figure \ref{fig:tatami}.  If corners such as these, where a $+$ occurs,
are forbidden but corners such as the upper right one in the \texttt{I} are allowed (a $\perp$ shape
or one of its rotations), then the nature of tatami coverings changes.  The complexity
of such coverings is unknown.

\subsubsection*{Acknowledgements}
Thanks to Bruce Kapron for useful conversations about this problem,
and Don Knuth for providing advice on a pre-print.  Part of this
research was conducted at the 9th McGill--INRIA Workshop on
Computational Geometry.

\bibliographystyle{abbrv}
\bibliography{bibliography}

\end{document}